\documentclass[3p,times,11pt]{elsarticle}
\usepackage{pifont}
\usepackage{amsmath, amsthm, amssymb, amsfonts}
\usepackage{graphicx}
\usepackage{geometry}
\usepackage{txfonts}
\usepackage[colorlinks]{hyperref}
\usepackage{natbib}
\usepackage[T1]{fontenc}
\usepackage[all]{xy}
\usepackage{times}
\usepackage{yfonts}
\usepackage[utopia]{mathdesign}
\usepackage{caption}
\usepackage{algorithm}
\usepackage[noend]{algorithmic} % noend is for for no endfor and no endif
\algsetup{indent=3mm} % controls the indentation
\newtheorem{thm}{Theorem}
\newtheorem{example}{Example}[section]
\newtheorem{remark}{Remark}
\newtheorem{defn}{Definition}[section]
\newtheorem{lemma}{Lemma}
\newtheorem{prop}{Proposition}
\newtheorem{cor}{Corollary}

\journal{submit to Discrete Mathematics}

\usepackage{amssymb}

\usepackage[figuresright]{rotating}

\begin{document}

\begin{frontmatter}

\title{On the $\ell$-DLIPs of codes over finite commutative rings}


\author[sb]{Sanjit Bhowmick}\ead{sanjitbhowmich392@gmail.com}
\author[fta]{Alexandre Fotue Tabue}\ead{alexfotue@gmail.com}
\author[jp]{Joydeb Pal\corref{cor1}} \ead{joydeb.palfma@kiit.ac.in}\cortext[cor1]{Corresponding author.}

\address[sb]{Department of Mathematics, National Institute of Technology  Durgapur, West Bengal, India.}
\address[fta]{Department of Mathematics, HTTC Bertoua, University of Bertoua, Bertoua, Cameroon.}
\address[jp]{Department of Mathematics,  School of Applied Sciences, Kalinga Institute of Industrial Technology (KIIT), Deemed to be University, Odisha, India.}

\begin{abstract}
Generalizing the linear complementary duals, the linear
complementary pairs and the hull of codes, we introduce the concept of $\ell$-dimension linear intersection pairs ($\ell$-DLIPs) of codes over a finite commutative ring $(R)$, for some positive integer $\ell$. In this paper, we study $\ell$-DLIP of codes over $R$ in a very general setting by a uniform method. Besides, we provide a necessary and sufficient condition for the existence of a non-free (or free) $\ell$-DLIP of codes over a finite commutative Frobenius ring. In addition, we obtain a generator set
of the intersection of two constacyclic codes over a finite chain ring, which helps us to get an important characterization of $\ell$-DLIP of constacyclic codes. Finally, the $\ell$-DLIP of constacyclic codes over a finite chain ring are used to construct new entanglement-assisted quantum error correcting (EAQEC) codes.
\end{abstract}


\end{frontmatter}

\section{Introduction}

$\ell$-linear intersection pairs ($\ell$-LIPs) of codes over finite fields have been studied due to their wide applications in cryptography. A pair of linear codes over finite fields is called an $\ell$-LIP if their intersection has dimension $\ell$. This idea was introduced by Guenda et al. \cite{Guenda2019}. The authors in \cite{Guenda2019} extended the concept of the linear complementary dual (LCD) codes and the linear complementary pair (LCP) of codes over finite fields (details of LCD codes, LCP of codes and their applications can be found in \cite{Bhasin2015, Bringer2014, Carlet2018, Massey1992}). In the same paper, Guenda et al. have shown that good $\ell$-LIP of codes over finite fields
exist and provided an application of linear $\ell$-intersection pair of codes by constructing entanglement-assisted quantum error correcting (EAQEC) codes.

%A linear $\ell$-intersection pair of codes over a finite chain ring is observed as a generalization of LCD and LCP codes over a finite chain ring. LCD codes were first introduced by Massey in 1992, where he established a necessary and sufficient condition for LCD codes over finite fields \cite{Massey1992}. Also, he constructed several asymptotically good codes over finite fields.

%In \cite{Carlet2018}, Carlet et al. proposed a structure of LCP codes over finite fields and constructed good example LCP codes, which is more useful than LCD codes over finite fields. But this topic has been popular for their practical application in the context of masking schemes and robustness against side channel attacks and fault injection attacks, which are shown in \cite{Bhasin2015, Bringer2014}. Guenda et al. extended in \cite{Guenda2019}, the study of LCP codes over finite fields to $\ell$-linear intersection pair ($\ell$-LIP) of codes over finite fields.

In \cite{LH20}, Liu and Hu defined a pair of codes to be an $\ell$-LIP of codes over finite chain rings if $\ell$ is the rank of their intersection. The notions of the rank of linear codes and the dimension of linear codes over finite fields coincide. Dougherty and Liu defined the rank of linear codes over finite rings in \cite{DL09}. But the rank of linear codes over finite rings is not a tool that allows us to determine the cardinality of these linear codes. These inspiring works highlighted the research gap in generalizing the dimension for linear codes over finite local rings.

In this paper, we introduce the dimension of linear codes over a finite local ring with the residue field $\mathbb{F}_q$ and in using the Chinese Remainder Theorem, we extend the dimension to linear codes over a finite commutative Frobenius ring. We generalize the $\ell$-LIP of codes over finite fields to $\ell$-dimension LIP ($\ell$-DLIP) of codes over finite local rings and study the characterizations of $\ell$-DLIP of codes over a finite commutative Frobenius ring.

The paper is organized as follows. In Section 2, we recall
background materials of linear codes over a finite commutative ring. In Section 3, we define $\ell$-DLIP of codes over a finite commutative ring and elaborate the property of $\ell$-DLIP of codes to obtain some salient characterizations for $\ell$-DLIP of codes over finite commutative Frobenius rings. Section 4 studies
$\ell$-DLIP of constacyclic codes over a finite chain ring.
Finally, we present some results on $\ell$-DLIP of codes under the Gray map and an application of $\ell$-DLIP of constacyclic codes over finite chain rings in constructing new entanglement-assisted quantum error correcting codes.

\section{Linear codes over a finite commutative ring}\label{sec:2}

This section uses $R$ as a finite commutative ring with
multiplicative unity $1$ ($0\neq 1$). The ring $R$ is a
principal ideal ring (PIR) if a single element generates each of its ideals. The ring $R$ is local if it has a unique
maximal ideal. In \cite[Theorem 2.2]{Nec08}, Nechaev demonstrated that the cardinality of any finite local ring is the power of its residue field.

\begin{lemma} \cite[Theorem 2.2]{Nec08} Let $R$ be a finite local ring whose residue field is $\mathbb{F}_q$ and $\textgoth{m}$ be a maximal ideal  of $R$. Then there is a positive integer $t$ (so-called the nilpotency index of $\textgoth{m}$) such that $R$ contains a strictly descending chain of the ideals $$R \supsetneq \textgoth{m} \supsetneq \textgoth{m}^2 \supsetneq \cdots \supsetneq \textgoth{m}^{t-1} \supsetneq \textgoth{m}^t=\{0\},$$ which satisfies the conditions $ t\leq \omega, |R|=q^\omega$ and $|\textgoth{m}|=q^{\omega-1}.$
\end{lemma}

For all $1\leq i <t$, the quotients $\textgoth{m}^i/\textgoth{m}^{i+1}$ are spaces over the field $\mathbb{F}_q$, and $\mu_i:=\dim_{\mathbb{F}_q}(\textgoth{m}^i/\textgoth{m}^{i+1})$. The parameters $(\mu_0, \cdots, \mu_{t-1})$ are called the \emph{Loewy invariants} of a local finite ring $R$. We have
$$|\textgoth{m}^i|=q^{\mu_i+\cdots+\mu_{t-1}}$$ for any $i
\in\{1,\cdots, t - 1\}$. Denote the maximal ideals of $R$ as $\textgoth{m}_1,\cdots,\textgoth{m}_u$. For any $1\leq i\leq u$, the stationary index $s_i$ of $\textgoth{m}_i$ is defined as $s_i:=\min\{k\in\mathbb{N}\;:\;\textgoth{m}_i^{k}=\textgoth{m}_i^{k+1}\}$. Clearly, $R_j:=R/\textgoth{m}_j^{s_j}$ is a finite local ring with maximal ideal $\textgoth{m}_j/\textgoth{m}_j^{s_j}$. Note that $s_i$ is the nilpotency index of $\textgoth{m}_j/\textgoth{m}_j^{s_j}$. Denote $\mathbb{F}_{q_j}$, the residue field of $R_j$. Then we have the ring epimorphisms
\begin{align}
\begin{array}{cccc}
 \Phi_j: & R & \rightarrow & R_j \\
    & a & \mapsto & a+\textgoth{m}_j
\end{array}
\end{align} and $\texttt{Ker}(\Phi_j)=\textgoth{m}_j^{s_j},$  for $1\leq j\leq u$. The ideals $\textgoth{m}_1,\cdots,\textgoth{m}_u$ are coprime and $\bigcap_{j=1}^u \textgoth{m}_j^{s_j} = \{0_R\}$. Using Chinese
remainder theorem (see \cite[p.224]{McD74}), the ring epimorphisms $\Phi_j$ ($1\leq j\leq u$) induce the following ring isomorphism as follows.

\begin{align}\label{Phi1}
\begin{array}{cccc}
 \Phi: & R & \rightarrow & R_1\times\cdots\times R_u \\
    & a & \mapsto & \left(\Phi_1(a),\cdots,\Phi_u(a)\right).
\end{array}
\end{align} The inverse of the map (\ref{Phi1}) is denoted as $\texttt{CRT}$, and $R$ is called the \emph{Chinese product of rings} $\{R_j\}_{j=1}^u.$

Recall that the Jacobson radical of $R$, denoted as $\texttt{J}(R)$, is the intersection of all maximal ideals of $R$. Also, the socle of $R$, denoted as $\texttt{Soc}(R)$, is the sum of the minimal $R$-submodules of $R$. The ring $R$ is Frobenius if the $R$-module $R$ is injective. Alternatively, $R$ is Frobenius if $R/\texttt{J}(R)\simeq \texttt{soc}(R)$ (as $R$-modules). Any finite PIR (so-called Chinese product of finite chain rings) is a Frobenius ring.

\begin{example} \label{ex2.1}
Let $\Re_k := \mathbb{F}_2[u_1,u_2,\cdots,u_k]$ be a finite commutative ring such that $$\{u_1,u_2,\cdots,u_k\}\cap \mathbb{F}_2=\emptyset\text{ and }u_1^2=u_2^2=\cdots=u_k^2=0,$$ where $k$ is a positive integer. For $\emptyset\neq A\subseteq\{1, 2, \cdots, k\}$, we set $u_A:=\prod\limits_{i\in A}u_i$ and $u_\emptyset:=1$. The ring $\Re$ is an $\mathbb{F}_2$-vector space with basis $\left\{u_A\;:\; A\subseteq\{1, 2, \cdots, k\}\right\}$. The ideals of $\Re_k$ are subspaces of $\Re_k$ whose basis are subsets of $\left\{u_A:\; \emptyset\neq A\subseteq\{1, 2, \cdots, k\}\right\}$. Thus, $\texttt{J}(\Re_k)=\langle u_1,u_2,\cdots,u_k\rangle$ and $\texttt{Soc}(\Re_k)=\left\langle\prod\limits_{i=1}^k u_i\right\rangle$. The Loewy invariants of $\Re_k$ is $(\mu_0, \mu_1, \cdots, \mu_k)$, where $\mu_i=\dim_{\mathbb{F}_2}\left(\texttt{J}(\Re_k)^i/\texttt{J}(\Re_k)^{i+1}\right)=
\left(\begin{array}{c}
  k \\
  i
\end{array}
\right)$ and $\Re_k/\texttt{J}(\Re_k) \simeq \texttt{Soc}(\Re_k)$ (as $\Re_k$-modules). Therefore, $\Re_k$ is a finite commutative local non-chain Frobenius ring with $\omega=2^k$.
\end{example}

\begin{lemma} \cite[p.224]{McD74} Any commutative Frobenius ring can be written as the Chinese product of commutative local Frobenius rings.
\end{lemma}

We use this decomposition of rings to understand codes
defined over finite commutative Frobenius rings. Denote the zero vector in $R^n$ as $\textbf{0}$. A linear code $C$ of length $n$ over a finite ring $R$ is an $R$-submodule of $R^n.$ For any linear code $C$ over $R$ of length $n$, the minimal generating sets for $C$ have the same cardinality, called the rank of $C$ and denoted as $\texttt{Rk}_R(C)$. Thus, $$\texttt{Rk}_R(C):= \texttt{min}\left\{i\in\mathbb{N}:\; \text{there exists a monomorphism $C \hookrightarrow R^i$ as $R$-modules}\right\}.$$ Let $\{\textbf{v}_1, \textbf{v}_2, \cdots, \textbf{v}_s\}$ be a set of non-zero vectors in $R^n$. Denote $\langle\,\textbf{v}_1, \textbf{v}_2, \cdots,
\textbf{v}_s\,\rangle$ as the set of all linear combinations of vectors $\textbf{v}_1, \textbf{v}_2, \cdots, \textbf{v}_s$. In \cite{DL09}, Dougherty and Liu studied several types of independence of vectors. Recall that the set $\{\textbf{v}_1, \textbf{v}_2, \cdots, \textbf{v}_s\}$ is
\begin{itemize}
\item  \emph{$R$-modular independent} if for any $(\alpha_1,\cdots, \alpha_s)$ in $R^s$, \;$\sum\limits_{i=1}^s\alpha_i\textbf{v}_i=\textbf{0}$ implies $\alpha_i\in\textgoth{m}$, for any $1\leq i\leq s$.
\item \emph{$R$-independent} if for some $(\alpha_1,\cdots, \alpha_s)$ in $R^s$, \;$\sum\limits_{i=1}^s\alpha_i\textbf{v}_i=\textbf{0}$ implies $\alpha_i\textbf{v}_i=\textbf{0}$, for some $1\leq i\leq s$.
\end{itemize}
The $R$-independence implies the $R$-modular independence when $R$ is local. Thus, a generator matrix for a linear code $C$ is a matrix whose rows form a basis for $C$. The
authors in \cite[Theorem 4.7]{DL09} proved that those bases exist for any code over a principal ideal ring and that the cardinality of any basis is equal to the rank $\texttt{Rk}_R(C)$.

\begin{example} We consider the ring cited in Example \ref{ex2.1}. Let $s$ be a nonnegative integer such that $s\leq k$ and $\{i_1,\cdots, i_s\}$ be a subset of $\{1, 2, \cdots, k\}$. The set of vectors $\{u_{i_1}, u_{i_2},\cdots, u_{i_s}  \}$ are $\Re_k$-modular independent, but they are not $\Re_k$-independent. Thus, the ideal $\langle u_{i_1}, u_{i_2},\cdots, u_{i_s} \rangle$ of $\Re_k$ is a linear code over $\Re_k$ of length $1$ with generator matrix 
$\left(
\begin{array}{c}
  u_{i_1} \\
  \vdots\\
  u_{i_s}
\end{array}
\right).$ Moreover, $|\langle u_{i_1}, u_{i_2}, \cdots,  u_{i_s} \rangle|=2^{\left(2^k-2^{k-s}\right)}$.
\end{example}

Let $C_j$ be an $n$-length code over $R_j$. We extend the map $\Phi$ coordinatewise to $R^n$ as
\begin{align}
\begin{array}{cccc}
 \Phi: & R^n & \rightarrow & (R_1)^n\times\cdots\times(R_u)^n \\
    & \textbf{a} & \mapsto &
    \left(\Phi_1(\textbf{a}),\cdots,\Phi_u(\textbf{a})\right),
\end{array}
\end{align}
where
\begin{align}\label{Phi}
\begin{array}{cccc}
 \Phi_j: & R^n & \rightarrow &  (R_j)^n \\  & (a_1,a_2,\cdots,a_n) & \mapsto &
 \left(\Phi_j(a_1),\Phi_j(a_2),\cdots,\Phi_j(a_n)\right).
\end{array}
\end{align}
$C$ is a linear code over $R$ if and only if for all
$1\leq j\leq u$,  $\Phi_j(C)$ is a linear code over $R_j$. Since $\Phi$ is bijective, $C = \texttt{CRT}(C_1,C_2, \cdots , C_u) := \Phi^{-1}(C_1 \times C_2 \times \cdots \times C_u),$ where $\Phi_j(C)=C_j$ for $1\leq j\leq u.$ In this case, $C$ is called the Chinese product of codes $\{C_j\}_{j=1}^u.$

\begin{lemma}\cite[Theorem 2.4]{DL09}\label{lem00} Let $C$ be the Chinese product of linear codes $\{C_j\}_{j=1}^n,$ where $C_j$ is a linear code over $R_j$, and $\mathbb{F}_{q_j}$ is the residue field of $R_j$. Then
\begin{enumerate}
    \item $|C| = \prod_{i=1}^u|C_j|;$
    \item $\texttt{Rk}_R(C)=\texttt{max}\left\{\texttt{Rk}_{R_j}(C_j):\;1\leq j\leq u\right\};$
    \item $C$ is a free code if and only if each $C_j$ is free with the same rank  $\texttt{Rk}_R(C).$
\end{enumerate}
\end{lemma}

We Denote $\dim(C_j):=\log_{q_j}(|C_j|)$ and $\texttt{Rank}_{q_j}(\mathrm{A}_j):=\log_{q_j}(|C_j|),$ where $\mathrm{A}_j$ is a matrix whose rows span $C_j$. By setting $Q:=\prod\limits_{j=1}^{u}q_j$ from Lemma \ref{lem00}, it follows
that
\begin{align}\label{Eq0}\log_Q(|C|)=\sum\limits_{j=1}^{u}\left(\frac{\ln(q_j)}{\ln(Q)}\right)\dim(C_j).
\end{align}
The positive integer $\log_Q(|C|)$ is called the dimension of $C$, denoted as $\dim(C)$. Now, we consider the standard inner product to the space $R^n$ as follows.
\begin{align}[\textbf{v}\,,\,\textbf{w}] = \sum_{j=1}^n v_jw_j,
\end{align} 
where $\textbf{v}:=(v_1,v_2,\cdots,v_n)$ and $\textbf{w}:=(w_1,w_2,\cdots,w_n)$ are vectors in $R^n.$  For a code $C,$ its dual code is defined as $C^\perp = \{\textbf{u}\in R^n:\; [\textbf{u}, \textbf{c}] = 0_R, \text{ for all } \;\textbf{c}\; \text{ in } C\}.$  A generator matrix for $C^\perp$
is called a parity-check matrix for $C$. It is well known that
for codes over Frobenius rings, $|C|\times|C^\perp| = |R|^n$ (see
\cite{Woo99} for the proof). Thus, for any linear code $C$ over $R$, we have
\begin{align}\label{eq*}
    \dim(C)+\dim(C^\perp)=n\omega.
\end{align}

\begin{example} We consider the ring cited in Example \ref{ex2.1}. Let $s$ be a nonnegative integer such that $s\leq k$ and $\{i_1,\cdots, i_s\}$ be a subset of $\{1, 2, \cdots, k\}$. Then in $\Re_k$, we have
$$\langle u_{i_1}, u_{i_2},\cdots, u_{i_s}
\rangle^\perp=\texttt{Ann}(\langle u_{i_1}, u_{i_2},\cdots,
u_{i_s} \rangle)=\left\langle\prod\limits_{j=1}^su_{i_j}\right\rangle.$$ Moreover, $\left|\left\langle\prod\limits_{j=1}^su_{i_j}\right\rangle\right|=2^{\left(2^{k-s}\right)}$.
\end{example}

%In the sequel of this paper, all rings will be Frobenius.

%%%%%%%%%%%%%%%%%%%%%%%%%%%%%%%%
\section{Characterization of $\ell$-DLIP of codes over a finite local Frobenius ring}

In this section, $R$ is a finite commutative local Frobenius ring with residue field $\mathbb{F}_q$, where $q$ is a prime power and $\omega$ is a positive integer such that $|R|=q^\omega$.

\begin{defn} \label{def 2.1}
For a nonnegative integer $\ell$, a pair $\{C, D\}$ of linear codes of length $n$ over $R$ is called an $\ell$-dimension linear intersection pair ($\ell$-DLIP) if $\dim(C\cap D)=\ell.$
\end{defn}

Let $\Re$ be a finite Frobenius ring such $\Re:=\texttt{CRT}(R_1, R_2, \cdots, R_s)$. From Eq.\,\ref{Eq0}, we extend the definition of the dimension of linear codes over finite local Frobenius rings to finite Frobenius rings. Therefore, the definition of $\ell$-DLIP of codes over $\Re$ is given as follows: a pair $\{C, D\}$ of linear codes of length $n$ over a finite Frobenius ring $\Re$ is an $\ell$-DLIP if $\dim(C\cap D)=\ell$. Using CRT, we have the following result.

\begin{prop} \label{propo 2.1.2} Let $\Re$ be a finite Frobenius ring such that $\Re:=\texttt{CRT}(R_1, R_2, \cdots, R_s)$, and $\{C, D\}$ be a pair of linear codes over $R$ such that $C:=\texttt{CRT}(C_1, C_2, \cdots, C_s)$ and $D:=\texttt{CRT}(D_1, D_2, \cdots, D_s)$. If for each $1\leq i\leq s$, $\{C_i, D_i\}$ is an $\ell$-DLIP of codes over $R_i$, then $\{C, D\}$ is an $\ell$-DLIP of codes over $\Re$.
\end{prop}

From the above definition, the following facts are immediate.

\begin{enumerate}
\item[\textbf{Fact} 1:] a linear code $C$ is LCD if $\{C,
C^\perp\}$ is $0$-DLIP (for details information, see
\cite{Bhowmick});
\item[\textbf{Fact} 2:] a free linear $0$-DLIP $\{C, D\}$ with $\dim(C)+\dim(D)=n$ is an LCP (for details
information, see \cite{Liu2021});
\item[\textbf{Fact} 3:] the dimension of hull of any linear code $C$ is $\ell$ if and only if $\{C, C^\perp\}$ is $\ell$-DLIP (for details information about the the hull of cyclic codes over a finite chain ring, see \cite{Tal21}).
\end{enumerate}

Therefore, the generalization of LCD codes, LCP of codes, and the hull of codes can be viewed as $\ell$-DLIP codes. Now, we derive the properties of $\ell$-DLIP of codes over $R$ in terms of their generator and parity-check matrices. In this context, the following propositions are required.

 \begin{prop}\label{lm-1} Let $C$ and $D$ be linear codes over $R$ of length $n$. Then $$\dim(C+D)=\dim(C)+\dim(D)-\dim(C\cap D).$$
 \end{prop}

 \begin{proof} Consider the map
 $$
\begin{array}{cccc}
  \phi : & C\times D & \mapsto & C+D \\
    & (\textbf{x}, \textbf{y}) & \mapsto & \textbf{x}+ \textbf{y}.
\end{array}
 $$
 This map is an $R$-module homomorphism. By the definition of $C+D$, the map $\phi$ is surjective. Therefore, according to the First Isomorphism Theorem, $C\times D/\texttt{Ker}(\phi)\simeq  C+D$ (as $R$-modules). Since $C\cap D \simeq \texttt{Ker}(\phi)$ (as $R$-modules), it follows that $\frac{|C\times D|}{|C\cap D|}=|C+D|$. Thus, $\dim(C)+\dim(D)-\dim(C\cap D)=\dim(C+D)$.
 \end{proof}

\begin{prop}\label{lm-2} Let $C$ be a linear code over $R$ of length $n$ and $\mathrm{A}$ be a matrix over $R$
 with $n$ column(s). Then $$\dim(\{\textbf{x}\in C:\;
 A\textbf{x}^\top=\textbf{0}\})=\dim(C)-\texttt{Rank}_q(\mathrm{A}).$$
 \end{prop}

 \begin{proof}
 We define a map
$$\begin{array}{cccc}
  \psi: & C & \rightarrow & \texttt{Im}(\psi) \\
    & \textbf{x} & \mapsto & \mathrm{A}\textbf{x}^\top,
\end{array}$$
which is an $R$-module epimorphism. Since $\texttt{ker}(\psi)=\{\textbf{x}\in C\;:\; \mathrm{A}\textbf{x}^\top=\textbf{0}\}$, it is easy to verify that $\texttt{ker}(\psi)$ and the solution space of $\mathrm{A}\textbf{x}^\top=\textbf{0}$ both are $R$-module isomorphic. By the First Isomorphism Theorem, $C/\texttt{ker}(\psi)\simeq \texttt{Im}(\psi)$ (as $R$-modules). Thus, $\dfrac{\mid C\mid}{\mid \texttt{ker}(\phi)\mid}=|\texttt{Im}(\psi)|$. Therefore, $$\log_q(\mid C\mid)-\log_q(\mid \texttt{ker}(\psi)\mid)=\log_q(|\texttt{Im}(\psi)|).$$ From the definition of the dimension of a linear code, we have $\dim(C)-\dim(\texttt{ker}(\psi))=\dim(\texttt{Im}(\psi))$.
 \end{proof}

\begin{prop}\label{th1.00} Let $\{C, D\}$ be a pair of linear codes of length $n$ over $R$. Then \begin{align}\dim(C^\perp\cap D^\perp)=n\omega-\dim(C)-\dim(D) +\dim(C\cap D).\end{align}
\end{prop}

\begin{proof} It is well known that $C^\perp+D^\perp=(C\cap D)^\perp.$ From Proposition \ref{lm-1}, we obtain
$$\dim(C^\perp+D^\perp)=\dim(C^\perp)+\dim(D^\perp)-\dim(C^\perp\cap D^\perp).$$ Besides, Eq.\ref{eq*} gives $\dim((C\cap D)^\perp)=n\omega-\dim(C\cap D).$ This means $\dim(C^\perp\cap D^\perp)=n\omega-\dim(C)-\dim(D) +\dim(C\cap D)$.
\end{proof}

To establish the necessary and sufficient condition of $\ell$-DLIP of codes, the following lemmas are needed.

\begin{lemma}\label{th2.00}
Let $C_i$ be a linear code of length $n$ over $R$ with generator matrix $\mathrm{G}_i$, for $i=1,2$. Then $$\texttt{Rank}_q \left(
{\begin{array}{ccc}
   \mathrm{G}_1 \\
   \mathrm{G}_2 \\
   \end{array} } \right)=\texttt{Rank}_q(\mathrm{G}_1)+\texttt{Rank}_q(\mathrm{G}_2)-\dim(C_1\cap C_2).$$
\end{lemma}

\begin{proof} We set $\dim(C_1\cap C_2)=\ell.$ We know that $x \in C_1^\perp$ if and only if $\mathrm{G}_1x^\top={\bf 0}$ and $x \in C_2^\perp$ if and only if $\mathrm{G}_2x^\top={\bf 0}$. Thus, $x \in C_1^\perp \cap C_2^\perp$ if and only if $\mathrm{B}x^\top={\bf 0},$ where
$\mathrm{B} = \left( {\begin{array}{ccc}
   \mathrm{G}_1 \\
   \mathrm{G}_2 \\
   \end{array} } \right)$.
From Proposition \ref{lm-2}, we obtain $\dim\{\textbf{x}\in R^n\;:\; \mathrm{B}\textbf{x}^\top=\textbf{0}\}=n\omega-\texttt{Rank}_q(\mathrm{B}).$ Combining it with Proposition \ref{th1.00}, we get $n\omega-\dim(C_1)-\dim(C_2)+\ell=n\omega-\texttt{Rank}_q(B),$ which means
$\texttt{Rank}_q(\mathrm{B})=\texttt{Rank}_q(\mathrm{G}_1)+\texttt{Rank}_q(\mathrm{G}_2)-\ell$,
as $\dim(C_i)=\texttt{Rank}_q(\mathrm{G}_i)$. Therefore,
$$\texttt{Rank}_q \left( {\begin{array}{ccc}
   \mathrm{G}_1 \\
   \mathrm{G}_2 \\
   \end{array} } \right)=\texttt{Rank}_q(\mathrm{G}_1)+\texttt{Rank}_q(\mathrm{G}_2)-\dim(C_1\cap C_2).$$
\end{proof}

Combining Definition \ref{def 2.1} with Lemma \ref{th2.00}, the two results are immediate consequences.

\begin{remark} Let $C_i$ be a linear code of length $n$ over $R$ with generator matrix $\mathrm{G}_i$, for $i=1,2$. Then $\{C_1, C_2\}$ is an $\ell$-DLIP if and only if
$$\texttt{Rank}_q \left( {\begin{array}{ccc}
   \mathrm{G}_1 \\
   \mathrm{G}_2 \\
   \end{array} } \right)=\texttt{Rank}_q(\mathrm{G}_1)+\texttt{Rank}_q(\mathrm{G}_2)-\ell.$$
\end{remark}

\begin{remark}\label{th2.10}
Let $C_i$ be a linear code of length $n$ over $R$ with parity check matrix $\mathrm{H}_i$, for $i=1,2$. Then
$$\dim(C_1^\perp\cap C_2^\perp)=\texttt{Rank}_q(\mathrm{H}_1)+\texttt{Rank}_q(\mathrm{H}_2)-\texttt{Rank}_q \left(
{\begin{array}{ccc}
   \mathrm{H}_1 \\
   \mathrm{H}_2 \\
   \end{array} } \right).$$
\end{remark}

\begin{lemma}\label{lm2}
Let $C_i$ be a linear code of length $n$ over $R$ with generator matrix $\mathrm{G}_i$ and parity check matrix $\mathrm{H}_i$, for $i=1,2$. Then
$$\dim(C_1\cap C_2)=\dim(C_1)-\texttt{Rank}_q(\mathrm{H}_2\mathrm{G}_1^\top)$$ or
$$\dim(C_1\cap C_2)=\dim(C_2)-\texttt{Rank}_q(\mathrm{H}_1\mathrm{G}_2^\top).$$
\end{lemma}

\begin{proof} We set $d_i:=\texttt{Rk}_R(C_i).$ Then there is a submodule $\textgoth{M}$ of $R^{d_1}$ such that
$$
\begin{array}{cccc}
  \phi : & \textgoth{M} & \rightarrow &  C_1 \\
    & \textbf{m} & \mapsto & \textbf{m}\mathrm{G}_1
\end{array}
$$
is an $R$-module isomorphism. $x\in C_1\cap C_2$ if and only if there is $(\alpha, \beta)$ in $\textgoth{M}\times R^{d_2}$ such that $x=\alpha \mathrm{G}_1=\beta \mathrm{G}_2.$ Therefore, $\mathrm{H}_2x^\top=\mathrm{H}_2(\alpha
\mathrm{G}_1)^\top=(\mathrm{H}_2\mathrm{G}_1^\top)\alpha^\top$ and $\mathrm{H}_2x^\top=\mathrm{H}_2(\beta
\mathrm{G}_2)^\top=(\mathrm{H}_2\mathrm{G}_2^\top)\beta^\top=\textbf{0}$, since $\mathrm{H}_2\mathrm{G}_2^\top=\textbf{O}$. Thus, $C_1 \cap C_2=\{\alpha \mathrm{G}_1:\;(\exists \alpha\in \textgoth{M})((\mathrm{H}_2\mathrm{G}_1^\top)\alpha^\top={\bf 0})\}$. Since $\phi$ is an isomorphism, $C_1\cap C_2$ and
$$\{\alpha \in
\textgoth{M}\;:\;(\mathrm{H}_2\mathrm{G}_1^\top)\alpha^\top={\bf 0})\}$$ are isomorphic. Therefore, from Proposition \ref{lm-2} we have
$$\dim(C_1\cap C_2)=\dim(\{\alpha \in
\textgoth{M}:\;(\mathrm{H}_2\mathrm{G}_1^\top)\alpha^\top={\bf
0})\})=\dim(C_1)-\texttt{Rank}_q(\mathrm{H}_2\mathrm{G}_1^\top).$$
%A similar proof can be followed for the alternative part.
\end{proof}

Next, we give an essential characterization of $\ell$-DLIP of codes which is a straightforward application of Lemma \ref{lm2}.

\begin{thm}\label{th2}
Let $C_i$ be a linear code of length $n$ over $R$ with generator
matrix $\mathrm{G}_i$ and parity check matrix $\mathrm{H}_i$, for
$i=1,2$. Then $\{C_1, C_2\}$ is an $\ell$-DLIP if and only if
$\texttt{Rank}_q(\mathrm{H}_2\mathrm{G}_1^\top)=\texttt{Rank}_q(\mathrm{G}_1)-\ell$
or
$\texttt{Rank}_q(\mathrm{H}_1\mathrm{G}_2^\top)=\texttt{Rank}_q(\mathrm{G}_2)-\ell.$
\end{thm}

\begin{example} In $\Re_3:=\mathbb{F}_2[u_1, u_2, u_3]$,
consider $C_1:=\langle u_2, u_3\rangle$ and $C_2:=\langle u_1, u_3\rangle$. Then $$\mathrm{G}_1:=\left(%
\begin{array}{c}
  u_2 \\
  u_3
\end{array}%
\right), \mathrm{H}_1:=\left(u_2u_3\right),\,\mathrm{G}_2:=\left(%
\begin{array}{c}
  u_1 \\
  u_3
\end{array}%
\right),  \text{ and } \mathrm{H}_2:=\left(u_1u_3\right).$$ Thus, $\dim(C_1)=\dim(C_2)=2^3-2=6$. Now, $\mathrm{H}_2\mathrm{G}_1^\top=\mathrm{H}_1\mathrm{G}_2^\top=\left(%
\begin{array}{cc}
  u_1u_2u_3 & 0
\end{array}%
\right).$ Since $\langle
u_1u_2u_3\rangle=\texttt{J}(\Re_3)^\perp$, then
$\texttt{Rank}_q(\mathrm{H}_1\mathrm{G}_2^\top)=\texttt{Rank}_q(\mathrm{H}_2\mathrm{G}_1^\top)=\log_2|\langle
u_1u_2u_3\rangle|=1$, and so $\dim(C_1\cap C_2)=6-1=5$. Indeed, $\{ u_2, u_3, u_1u_2, u_1u_3, u_2u_3, u_1u_2u_3\}
\text{ and } \{ u_1, u_3, u_1u_2, u_1u_3, u_2u_3, u_1u_2u_3\}$ are $\mathbb{F}_2$-bases of $C_1$ and $C_2$, respectively. It deduces $\{ u_3, u_1u_2, u_1u_3, u_2u_3, u_1u_2u_3\}$ is an $\mathbb{F}_2$-basis of $C_1\cap C_2$ and $\left(%
\begin{array}{c}
  u_3 \\
  u_1u_2
\end{array}
\right)$ is a generator matrix for $C_1\cap C_2$. Hence $\{C_1, C_2\}$ is $5$-DLIP.
\end{example}

\begin{cor}
 Let $C_i$ be a linear code of length $n$ over $R$ with generator matrix $\mathrm{G}_i$  and parity check matrix $\mathrm{H}_i$, for $i=1,2$. Then a pair $\{C_1, C_2\}$ of codes with $C_1+C_2=R^n$ is $\ell$-DLIP if and only if
$\texttt{Rank}_q(\mathrm{H}_2\mathrm{G}_1^\top)=\texttt{Rank}_q(\mathrm{G}_1)-\ell=n\omega-\texttt{Rank}_q(\mathrm{G}_2)$
or
$\texttt{Rank}_q(\mathrm{H}_1\mathrm{G}_2^\top)=\texttt{Rank}_q(\mathrm{G}_2)-\ell=n\omega-\texttt{Rank}_q(\mathrm{G}_1)$.
\end{cor}

\begin{proof}
Since $C_1+C_2=R^n$, it follows that $\dim(C_1+C_2)=n\omega$. From Proposition \ref{lm-1}, we obtain $\texttt{Rank}_q(\mathrm{G}_1)+\texttt{Rank}_q(\mathrm{G}_2)=n\omega+\ell$. Combining it with Lemma \ref{lm2} follows the proof of the theorem.
\end{proof}

In the case when $C_2=C_1^\perp$, we have $ \texttt{Hull}(C_1)=\texttt{Hull}(C_1^\perp)= C_1\cap C_1^\perp$. Thus, we have the following result.

\begin{cor}
Let $C$ be a linear code over $R$ of length $n$ with generator matrix $\mathrm{G}$ and parity check matrix $\mathrm{H}$. Then
$\dim(\texttt{Hull}(C))=\ell$ if and only if
$\texttt{Rank}_q(\mathrm{G}\mathrm{G}^\top)=\texttt{Rank}_q(\mathrm{G})-\ell$
or
$\texttt{Rank}_q(\mathrm{H}\mathrm{H}^\top)=\texttt{Rank}_q(\mathrm{H})-\ell.$
\end{cor}

Now, we recall the definition of the projective module as follows: an $R$-module $C$ is projective if there is an $R$-module $M$ such that $C\oplus M$ is a free $R$-module.

\begin{remark}\label{rem1}
Let $\mathrm{B_1}$ and $\mathrm{B_2}$ be $R$-modules. If
$\mathrm{B_1}\oplus \mathrm{B_2}$ is free, then $\mathrm{B_1}$ and $\mathrm{B_2}$ are projective.
\end{remark}

\begin{lemma}\cite[Theorem 2]{Kap58}\label{lem2}
Over chain rings, any projective module is free.
\end{lemma}

\begin{thm}\label{th2.1}
Let $C_i$ be a non-free linear code of length $n$ over $R$, for $i=1,2$ such that $C_1+C_2=R^n$. Then there always exists a non-zero $\ell$-DLIP of codes.
\end{thm}

\begin{proof}
Let $ C_1\cap C_2=\{\,\textbf{0}\,\}$. Then we have $C_1\oplus C_2=R^n$, as $C_1+C_2=R^n$. Since $C_1\oplus C_2$ is free, then Remark \ref{rem1} implies that $C_1$ as well as $C_2$ are projective. Now, Lemma \ref{lem2} claims that $C_1$ as well as $C_2$ must be free, which is a contradiction to the fact that $C_1$ and $C_2$ are non-free. Therefore, $C_1\cap C_2\neq\{\,\textbf{0}\,\}$.
\end{proof}

\begin{cor}\label{th2.1} Let $\{C_1, C_2\}$ be a pair of linear codes over $R$ of length $n$ such that $C_1+C_2=R^n$. If $\{C_1, C_2\}$ is $0$-DLIP, then $C_1$ and $C_2$ are free.
\end{cor}

%\begin{proof} Let $\{C_1,\,C_2\}$ be a $0$-DLIP of linear codes over $R$ of length $n$ such that $C_1+C_2=R^n$. Then $C_1\oplus C_2=R^n$ and hence $C_1$ and $C_2$ are projective, as $C_1\oplus C_2$ is free. Therefore, from Lemma \ref{lem2}, we conclude that $C_1$ and $C_2$ must be free.
%\end{proof}

By setting $\ell=0$, combining Lemma \ref{th2.00} and Corollary \ref{th2.1}, we have the following result.

\begin{cor}
Let $C_i$ be a free linear code of length $n$ over $R$  with
generator matrix $\mathrm{G}_i$ and parity check matrix
$\mathrm{H}_i,$ for $i=1,2$, such that $C_1+C_2=R^n$. Then the following statements are equivalent.
\begin{enumerate}
\item  $C_1$ and $C_2$ form LCP; \item  $\left(
{\begin{array}{ccc}
   \mathrm{G}_1 \\
   \mathrm{G}_2 \\
   \end{array} } \right)$ is invertible over $R$;
\item  $\left( {\begin{array}{ccc}
   \mathrm{H}_1 \\
   \mathrm{H}_2 \\
   \end{array} } \right)$ is invertible over $R$;
   \item $\texttt{Rank}_q(\mathrm{H}_2\mathrm{G}_1^\top)=\texttt{Rank}_q(\mathrm{G}_1)$
or
$\texttt{Rank}_q(\mathrm{H}_1\mathrm{G}_2^\top)=\texttt{Rank}_q(\mathrm{G}_2)$.
\end{enumerate}
\end{cor}

%%%%%%%%%%%%%%%%%%%%
\section{Constacyclic $\ell$-DLIP of codes over a finite chain ring}

In this section, $R$ is assumed to be a finite chain ring (FCR) of invariants $(q, t)$. Therefore, $R$ is a local ring with maximal ideal $\textgoth{m}$. Indeed, $\textgoth{m}$ is a principal ideal.
Thus, $\textgoth{m}=\langle \gamma\rangle,$ $R/\textgoth{m}\simeq \mathbb{F}_q$ and $\{0\}=\langle \gamma^t \rangle\subsetneq \langle \gamma^{t-1} \rangle \subsetneq \cdots \subsetneq \gamma \subsetneq \langle \gamma^0 \rangle =R$, where $t$ is called the nilpotency index of $\gamma$ and $\mathbb{F}_{q}$ is the residue field of $R$.

We introduce the mapping
$$\begin{array}{cccc}
  \pi: & R & \rightarrow & \mathbb{F}_q \\
    & r & \mapsto & r+ \textgoth{m},
\end{array}$$
which is a natural epimorphism from $R$ onto $\mathbb{F}_q.$ This mapping $\pi$ can be extended naturally to an epimorphism from $R^n$ to $\mathbb{F}_q^n$. For a given unit $\lambda$ in $R$, a linear code $C$ of length $n$ over $R$ is $\lambda$-constacyclic if $(\lambda c_{n-1}, c_0, c_1, \cdots, c_{n-2}) \in C$ whenever $(c_0, c_1,\cdots, c_{n-1}) \in C.$ The $1_R$-constacyclic codes are called cyclic codes and $(-1_R)$-constacyclic codes are called negacyclic codes. It is well known that the image of a $\lambda$-constacyclic code of length $n$ over $R$ under the $R$-isomorphism

\begin{align*}
\begin{array}{cccc}
  \Psi : & R^n & \rightarrow & R[X]/\langle X^n-\lambda\rangle \\
    & (c_0, c_1,\cdots, c_{n-1}) & \mapsto & c_0+ c_1x+\cdots+
    c_{n-1}x^{n-1},
\end{array}
\end{align*}
is an ideal of the residue class ring $\frac{R[x]}{\langle
x^n-\lambda\rangle}.$ Now, for each $a(x):=\sum\limits_{i=0}^s a_ix^i$ with $a_{0}$ invertible in $R$, the reciprocal of $a(x)$ is denoted by $a^*(x)$ and is defined by $a^*(x)=x^sa_0^{-1}\sum\limits_{i=0}^sa_ix^{-i}$. If
$a(x)=a^*(x)$, then $a(x)$ is called self-reciprocal. Now, we focus on $\lambda$-constacyclic $\ell$-DLIP of codes over $R$. Our next result on non-free $\lambda$-constacyclic over $R$ follows from \cite[Theorem 3.14]{GG12}.

\begin{lemma}\label{$l-intersection$, th5.1}
Let $C$ be a $\lambda$-constacyclic code of length $n$ over $R$. Then there exists a unique $(t+1)$-tuple $(F_0,\dots, F_t)$ of pairwise coprime polynomials in $R[x]$ satisfying $F_0\cdots F_t=x^n-\lambda$, such that
$$\Psi(C)=\left\langle\left\{\gamma^{t}\widehat{\,F_{j+1}}\;:\;0\leq j < t\right\}\right\rangle \text{~and~}\Psi(C^\perp)=\left\langle\left\{\widehat{\,F_0\;}^*\right\}\cup\left\{\gamma^{j}\widehat{\,F_{t-j+1}}^*\;:\;1\leq j \leq t-1\right\}\right\rangle,$$ where $\widehat{\,F_i\;}=\frac{x^n-\lambda}{F_i}$, for $0\leq i\leq t$.
 Moreover, $|C|= q^{\sum\limits_{j=0}^{t-1}{(t-j)}\deg(F_{j+1})}~\mbox{and}~|C^\perp|= q^{\sum\limits_{j=1}^{t}{j}\deg(F_{j+1})}$ with $F_{t+1}=F_0$.
\end{lemma}

%%%%%%%%%%%%%
Our next result about the intersection of two $\lambda$-constacyclic codes over $R$ follows from Lemma
\ref{$l-intersection$, th5.1}, the division algorithm theorem and the notion of least common multiple.

\begin{thm}\label{sec4thm}
Let $C_i$ be a $\lambda$-constacyclic code over $R$ such that $$\Psi(C_i)=\langle\{\gamma^j\widehat{\,F_{i,j}\;}\;:\;0\leq j <t\}\rangle,$$ where ${F}_{i,0}{F}_{i,1}\cdots{F}_{i,t}=x^n-\lambda$, for $i=1,2$. Then
\begin{align}
\Psi(C_1\cap
C_2)=\left\langle\left\{\gamma^{j}\texttt{lcm}(\widehat{\,F_{1,j}\;},\widehat{\,F_{2,j}\;})\;:\;
0\leq j<t\right\}\right\rangle.
\end{align}
Moreover,\; $\dim(C_1\cap
C_2)=\sum\limits_{j=0}^{t-1}(t-j)(n-\deg(\texttt{lcm}(\widehat{\,F_{1,j}\;},\widehat{\,F_{2,j}\;}))).$
\end{thm}

For better understanding, we discuss an example over the ring $\mathbb{Z}_4$ below.

\begin{example}
We know that $x^7+1=(x+1)(x^3-2x^2+x+1)(x^3+x^2+2x+1)$ over
$\mathbb{Z}_4$. Let $f(x)=x^3+x^2+2x+1$, $g(x)=x^3-2x^2+x+1$ and $h(x)=x+1$ which are basic irreducible polynomials. Let $C_1$ and $C_2$ be two negacyclic codes of length $7$ over $\mathbb{Z}_4$.
Notice that $$\Psi(C_1)=\langle f(x)g(x), 2f(x) \rangle\text{ and }\Psi(C_2)=\langle h(x)f(x), 2h(x) \rangle$$ are also negacyclic codes of length $7$ over $\mathbb{Z}_4$. Therefore,
\begin{eqnarray*}
   \Psi(C_1\cap C_2) &=&  \langle \texttt{lcm}(f(x)g(x), h(x)f(x)), 2\texttt{lcm}(f(x),
h(x)) \rangle; \\
    &=& \langle 2(x+1)(x^3+x^2+2x+1) \rangle,
\end{eqnarray*}
Therefore, $\{C_1, C_2\}$ is a $3$-DLIP of codes.
\end{example}

%%%%%%%%%%%%

\begin{cor}\label{sec4cor}
Let $C_i$ be a free $\lambda$-constacyclic code over $R$ such that $\Psi(C_i)=\langle F_i\rangle$, where $F_i$ divides $x^n-\lambda$, for $i=1,2$. Then  $C_1+C_2=R^n$ if and only if $\texttt{lcm}(F_1,F_2)=F_1F_2.$ Moreover, $C_1$ and $C_2$ form an LCP if and only if $\deg(F_1F_2)=n$.
\end{cor}

\begin{proof} We have $\dim(C_i)=t(n-\deg( F_i))$. From Proposition \ref{lm-1}, we get $$\dim(C_1+
C_2)=\dim(C_1)+\dim(C_2)-\dim(C_1\cap C_2).$$ Since $C_1+
C_2=R^n$, it follows that $\dim(C_1+ C_2)=tn$. Thus, $\dim(C_1\cap C_2)=t(n-\deg(F_1)-\deg(F_2))$. By Theorem \ref{sec4thm}, $\dim(C_1\cap C_2)=t(n-\deg(\texttt{lcm}(F_1,F_2)))$. Hence $\deg(\texttt{lcm}(F_1,F_2))=\deg(F_1)+\deg(F_2).$
\end{proof}

Let $C_1$ be a $\lambda_1$-constacyclic code and $C_2$ be a $\lambda_2$-constacyclic code, where $\pi{(\lambda_1)}\neq\pi{(\lambda_2)}$. With this assumption, we propose our next result.

\begin{prop}\label{$l-intersection$, lm5.1}
Let $C_i$ be a free $\lambda_i$-constacyclic code over $R$ of length $n$, for $i=1,2$, with $\pi{(\lambda_1)}\neq\pi{(\lambda_2)}$ such that $C_1\cap C_2$ is both $\lambda_1$ and $\lambda_2$-constacyclic. Then $C_1\cap
C_2= R^n$ or $C_1\cap C_2=\{\,\textbf{0}\,\}$.
\end{prop}

\begin{proof} Assume that $C_1\cap C_2\neq\{\,\textbf{0}\,\}$. Since the codes $C_1$ and $C_2$ are free, $C_1\cap C_2$ is also free. Thus, there exists $\textbf{x}:=(x_1, x_2, \cdots, x_n)$ in $C_1\cap C_2$ such that $x_1$ is an invertible element in $R$. Denote $\tau_i(\textbf{x})$ as
the $\lambda_i$-constashift of $\textbf{x}$. Then
$\tau_i(\textbf{x})=(\lambda_i x_n, x_1, \cdots, x_{n-1})$. Since $C_1\cap C_2$ is both $\lambda_1$ and $\lambda_2$-constacyclic, it follows that $\tau_1(\textbf{x})-\tau_2(\textbf{x})=((\lambda_1-\lambda_2)x_n,
0,\cdots, 0)\in C_1\cap C_2$. Since $\pi{(\lambda_1)}\neq\pi{(\lambda_2)}$, it deduces that $\lambda_1-\lambda_2$ is an invertible element in $R$, and so $(1,0,\cdots,0)\in C_1\cap C_2$. Hence $C_1\cap C_2= R^n$.
\end{proof}

\begin{remark} \label{rem4}
Let $\lambda_1$ and $\lambda_2$ be two invertible elements in $R$ such that $\pi{(\lambda_1)}\neq\pi{(\lambda_2)}$. Let $C_i$ be a free $\lambda_i$-constacyclic code over $R$ of length $n$, for $i=1,2,$ such that $C_1\cap C_2$ is both $\lambda_1$ and $\lambda_2$-constacyclic and $C_1+C_2=R^n$. Then $\{C_1, C_2\}$ is an LCP.
\end{remark}

If we set $C_1=C$ and $C_2=C^\perp$, then $C$ and $C^\perp$ form a free $\lambda$-constacyclic and free $\lambda^{-1}$-constacyclic codes over $R$, respectively. We have the following result in this direction using Corollary \ref{sec4cor} and Remark \ref{rem4}.

\begin{cor}
Let  $C$ be a free $\lambda$-constacyclic code over $R$ such that $\Psi(C)=\langle F\rangle$ where $F$ divides $ x^n-\lambda$, with $\pi(\lambda^2)=1$. Then $C$ forms an LCD code if and only if $F(x)$ is self-reciprocal.
\end{cor}

%\begin{proof}
%$C$ forms LCD over $R$, if and only if $C\cap C^\perp=\{\,\textbf{0}\,\}$, if and only if $C\cap C^\perp=\langle
%\texttt{lcm}(F,\widehat{F}^*)\rangle=0$, where $\Psi(C^\perp)=\langle \widehat{F}^*\rangle$ with $\widehat{F}=\frac{x^n-\lambda}{F}$, if and only if $F\widehat{F}^*=x^n-\lambda$, as $\gcd(F, \widehat{F}^*)=1$, for $\gcd(n,p)=1$, if and only if $F^*\widehat{F}^*=x^n-\lambda^{-1}=x^n-\lambda=F\widehat{F}^*$, as
%$\lambda^2=1$, if and only if $F=F^*$.
%\end{proof}

In order to calculate all $\lambda$-constacyclic LCD codes, we need to consider the case $\pi{(\lambda)}\neq\pi{(\lambda)}^{-1}$. In this regard, we state the following result from Proposition \ref{$l-intersection$, lm5.1}.

\begin{cor}
Let $C$ be a free $\lambda$-constacyclic code over $R$ of length $n$ such that $\Psi(C)=\langle F\rangle$, where $F$ divides $x^n-\lambda$ and $\pi{(\lambda)}\neq\pi{(\lambda)}^{-1}$. Then $C$ forms an LCD over $R$.
\end{cor}

In the sequel, we consider the ring $\mathbb{F}_q[\gamma]$ with $\gamma\neq\gamma^2=0$. Also, assume that $q=p^s$, where $p$ is a prime number, $s$ is a positive integer, and $q\equiv 1 \pmod 4$. Thus, $\mathbb{F}_{q}[\gamma]$ is a finite commutative chain ring and $\texttt{J}(\mathbb{F}_{q}[\gamma])=\gamma \mathbb{F}_{q}$. Note that $\mathbb{F}_{q}[\gamma]$ is a vector space over $\mathbb{F}_q$ with ordered basis $(1, \gamma)$. In \cite[Section 5]{LL20} the Gray map $\phi : \mathbb{F}_{q}[\gamma] \rightarrow \mathbb{F}_q^2$ has been defined as below.
$$
\begin{array}{cccc}
  \phi : & \mathbb{F}_{q}[\gamma] & \rightarrow & \mathbb{F}_q^2 \\
    & a+\gamma b & \mapsto & (\alpha b, a+b),
\end{array}
$$  where $\alpha^2=-1$, for some $\alpha$ in $\mathbb{F}_q[\gamma]$. This mapping is an isomorphism of vector spaces over $\mathbb{F}_q,$ which can be extended to $(\mathbb{F}_{q}[\gamma])^n \rightarrow (\mathbb{F}_q)^{2n}$ in the following way.
$$
\begin{array}{cccc}
  \phi : & (\mathbb{F}_{q}[\gamma])^n & \rightarrow & (\mathbb{F}_q)^{2n} \\
    & (x_1, x_2, \cdots, x_n) & \mapsto & (\phi(x_1), \phi(x_2), \cdots, \phi(x_n)).
\end{array}
$$
The map $\phi :   (\mathbb{F}_{q}[\gamma])^n   \rightarrow
(\mathbb{F}_q)^{2n}$ is also an isomorphism of  vector spaces over $\mathbb{F}_q$. Therefore, for any pair $\{C_1, C_2\}$ of linear codes over $\mathbb{F}_{q}[\gamma]$ of length $n$, we have $$\dim(C_1\cap C_2)=\dim(\phi(C_1)\cap\phi(C_2)).$$ This means that $\{C_1, C_2\}$ is an $\ell$-DLIP if and only if $\{\phi(C_1), \phi(C_2)\}$ is an $\ell$-LIP. Now, we define a Gray weight function over $\mathbb{F}_{q}[\gamma]$ as 
$$\texttt{wt}_G(a+\gamma b)=\left\{
\begin{array}{ll}
    0, & \hbox{if $a=b=0$;} \\
    2, & \hbox{if $b\neq 0$ and $a\neq -b$;} \\
    1, & \hbox{otherwise.}
\end{array}
\right.$$ The Gray weight over $\mathbb{F}_{q}[\gamma]$ can also be extended to $(\mathbb{F}_{q}[\gamma])^n$ as follows:
$$\texttt{wt}_G(x_1, x_2, \cdots, x_n):=\sum\limits_{i=1}^{n}\texttt{wt}_G(x_i),$$ for any $(x_1, x_2, \cdots, x_n)$ in $(\mathbb{F}_{q}[\gamma])^n$. The Gray map $\phi :
(\mathbb{F}_{q}[\gamma])^n\rightarrow(\mathbb{F}_q)^{2n}$ is an $\mathbb{F}_q$-linear map such that $\texttt{wt}_H(\phi(\textbf{x}))=\texttt{wt}_G(\textbf{x}),$ for
any $\textbf{x}$ in $(\mathbb{F}_{q}[\gamma])^n$, where
$\texttt{wt}_H$ is the Hamming weight on $(\mathbb{F}_q)^{2n}$. Therefore, the Gray image of any linear code $C$ over $\mathbb{F}_{q}[\gamma]$ of length $n$, is a linear code $\phi(C)$ over $\mathbb{F}_q$ of length $2n$ with $\dim(C)=\dim(\phi(C))$ and $\texttt{wt}_H(\phi(C))=\texttt{wt}_G(C)$.

Entanglement-assisted quantum error correcting (EAQEC) codes were introduced firstly by Hsieh et al. \cite{Hsich2007}. The authors in \cite{Hsich2007} have framed EAQEC codes from arbitrary classical linear codes. An $\left[\left[n, k, d; c\right]\right]_{q}$ EAQEC code over $\mathbb{F}_q$ encodes $k$ logical qudits into $n$ physical qudits with the help of $c$ copies of maximally entangled states. We recall the proposition \cite[Proposition 4.2]{Guenda2019}, which demonstrates that classical linear codes can be used to create EAQEC codes.

\begin{prop}\label{$l-intersection$, pro7.1}
Let $\{C_1, C_2\}$ be an $\ell$-LIP of codes with parameters
$[n,k_1,d_1]_q$ and $[n,k_2,d_2]_q$, respectively. Then there
exists an $[[n, k_2-\ell, \min\{d_1^\perp,d_2\}; k_1-\ell]]_q$
EAQEC code with $d_1^\perp=d(C_1^\perp)$.
\end{prop}

Using the Gray map, Proposition \ref{$l-intersection$, pro7.1} and
Theorem \ref{sec4thm}, we have the following corollary.

\begin{cor}
Let $C_i$ be a $\lambda$-constacyclic code over
$\mathbb{F}_{q}[\gamma]$ with minimum Gray weight $\texttt{w}_i$
such that
$\Psi(C_i)=\langle\{\widehat{\,F_{i,1}\;},\,\gamma\widehat{\,F_{i,1}\;}\}\rangle,$
where ${F}_{i,0}{F}_{i,1}{F}_{i,2}=x^n-\lambda$ (for
$i\in\{1,2\}$).
 Then  there exists an EAQEC code with parameter $$[[2n,
\tau_2, \min\{\texttt{w}_1^\perp, \texttt{w}_2\}; \tau_1]]_q,$$
where
$\tau_i=2\left(\deg(F_{i,1})+\deg(\texttt{lcm}(\widehat{\,F_{1,0}\;},
\widehat{\,F_{2,0}\;}))\right)+\left(\deg(F_{i,2})+\deg(\texttt{lcm}(\widehat{\,F_{1,1}\;},
\widehat{\,F_{2,1}\;}))\right)-3n,$ and
$\texttt{w}_1^\perp=\texttt{wt}_G(C_1^\perp).$
\end{cor}

\textbf{Acknowledgements}\; The authors would like to thank the anonymous referees for their careful reading, insightful comments, and suggestions which help us to improve our manuscript drastically.

%\section*{References}

\bibliographystyle{elsarticle-num}

\end{document}